\documentclass[twocolumn,journal]{IEEEtran}

\usepackage{graphics}
\usepackage{epstopdf}
\usepackage[pdftex]{graphicx}
\usepackage{stfloats}
\usepackage{array}
\usepackage{float}
\usepackage[cmex10]{amsmath}
\usepackage{amsmath,amsfonts,amssymb}
\usepackage{graphicx,graphics}
\usepackage{enumerate}
\usepackage{color}
\usepackage{textcomp}
\usepackage{verbatim}
\usepackage{amsthm}
\usepackage{multirow}
\usepackage{subcaption}
\usepackage{siunitx}
\usepackage{amsmath,amsthm}
\usepackage{cleveref}
\usepackage{cite}
\usepackage[font=footnotesize]{caption}
\usepackage[english]{babel}
\makeatletter
\renewcommand{\fps@figure}{htp}
\renewcommand{\fps@table}{htp}
\makeatother
\renewcommand\figurename{Fig.}
\addto\captionsenglish{\renewcommand{\figurename}{Fig.}}

\newtheorem{theorem}{Theorem}

\begin{document}
%
\title{A Spatial-Spectral Interference Model for Dense Finite-Area 5G mmWave Networks}
\author{Solmaz~Niknam,~\IEEEmembership{Student Member,~IEEE,}
        Balasubramaniam~Natarajan,~\IEEEmembership{Senior Member,~IEEE,}
        and~Reza~Barazideh,~\IEEEmembership{Student Member,~IEEE.}

\thanks{Authors are with the department of ECE, Kansas state university, Manhattan, KS (emails: \{slmzniknam, bala, rezabarazideh\}@ksu.edu).} 
}

%

\maketitle

\begin{abstract}
With the overcrowded sub-6 GHz bands, millimeter wave (mmWave) bands offer a promising alternative for the next generation wireless standard, i.e., 5G. However, the susceptibility of mmWave signals to severe pathloss and shadowing requires the use of highly directional antennas to overcome such adverse characteristics. Building a network with directional beams changes the interference behavior, since, narrow beams are vulnerable to blockages. Such sensitivity to blockages causes uncertainty in the active interfering node locations. Configuration uncertainty may also manifest in the spectral domain while applying dynamic channel and frequency assignment to support 5G applications. In this paper, we first propose a blockage model considering mmWave specifications. Subsequently, using the proposed blockage model, we derive a spatial-spectral interference model for dense finite-area 5G mmWave networks. The proposed interference model considers both spatial and spectral randomness in node configuration. Finally, the error performance of the network from an arbitrarily located user perspective is calculated in terms of bit error rate (BER) and outage probability metrics. The analytical results are validated via Monte-Carlo simulations. It is shown that considering mmWave specifications and also randomness in both spectral and spatial node configurations leads to a noticeably different interference profile.

\end{abstract}
\begin{figure*}[t]
\centering
\includegraphics[scale=0.5]{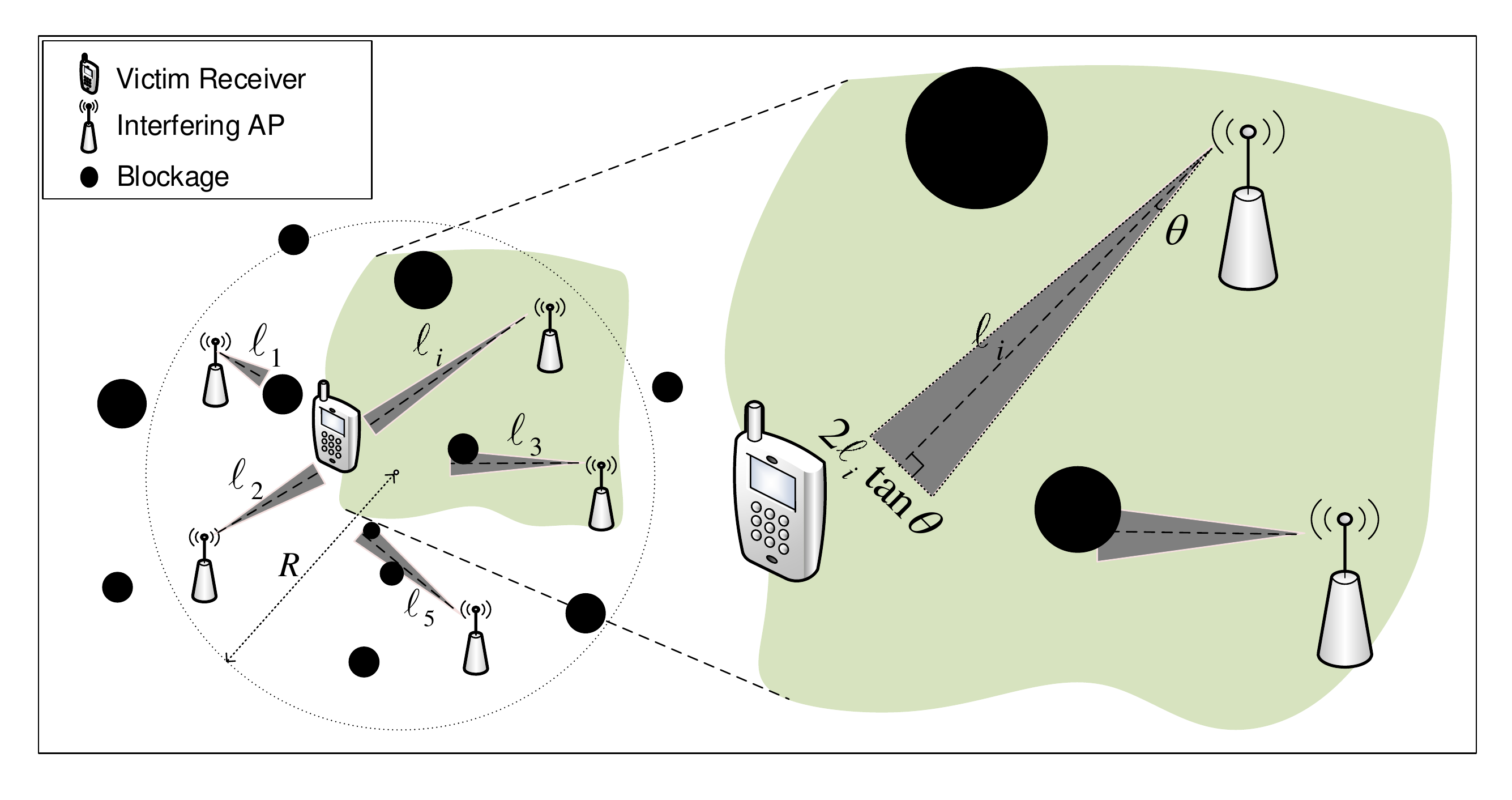}
\caption{The impact of interfering APs on the victim receiver in the presence of obstacles.}
\label{fig:blockage area}
\end{figure*}
\section{Introduction} \label{sec:intro}
Triggered by the popularity of smart devices, wireless traffic volume and device connectivity have been growing exponentially during recent years~\cite{Nokia2015ten}. Next generation of wireless networks, i.e., 5G, is a promising solution to satisfy the increasing data demand through combination of key enabling technologies such as ultra-densification (deployment of high density of access points (APs)) and utilization of large amount of bandwidth in millimeter wave (mmWave) bands. However, mmWave signals suffer from severe pathloss and strong gaseous attenuation due to oxygen molecules, water vapor, and rain drops in the atmosphere. Therefore, this section of spectrum has been under-utilized. However, having large antenna arrays that coherently direct the beam energy will help overcome the hostile characteristics of mmWave channels. Utilization of the highly directional beams changes many aspects of the wireless system design. One of the main factor that is highly impacted is the interference behavior.
In fact, directional links are highly susceptible to obstacles~\cite{MacCartneyJr2016} and interference in such networks tends to demonstrate an on-off pattern as a result of the movement of the nodes~\cite{Andrew2014what}. In addition to the node mobility, the unplanned user-installed APs which causes uncertainty in spatial network configuration influences the interference models as interference power is highly dependent on the relative locations of the transmitters and receivers. Therefore, it is important to consider random spatial models of the network configuration to accurately model the interference~\cite{Andrews2010spatial}. In addition to the spatial distance between nodes, the distance between their allocated frequencies also affects the amount of accumulated interference in a multi-user environment due to possible partial band overlap and out-of-band radiation. Therefore, uncertainty in network configuration may be observed in the spectral domain impacting the interference model. Moreover, with the notions of adaptive frequency selection and dynamic channel allocation strategies instead of static assignments~\cite{Ericsson2016Radio}, considering the uncertainty in spectral domain while modeling the accumulated interference provides a more accurate model. A precise interference model in turn impacts the design of interference coordination and management schemes.

\subsection{Related work}
There have been prior efforts on interference modeling considering random node distributions. Due to its analytical tractability, Poisson point process (PPP) is one of the popular random model assumed for node distribution~\cite{Haenggi2012stochBook}.
A 2D-PPP has been suggested in~\cite{Hamdi2009unified} in order to consider the presence of interferers in both spectral and spatial domains in the interference model. However, an infinite-sized network area and an unlimited frequency
band of operation are assumed in order to simplify the calculations. In addition, the effect of blockages and mmWave specifications such as high signal attenuation and beam directivity and are not taken into account. Therefore, it may not be applicable to 5G mmWave networks. Assuming a PPP model, an interference model for finite-sized highly dense mmWave networks has been proposed in~\cite{Heath2015wearable}. However, uncertainty in the configuration is only considered in the spatial domain. Moreover, Binomial point process (BPP) is a more appropriate choice for modeling finite-sized networks with a given number of APs~\cite{Haenggi2012stochBook}. Authors in~\cite{Dong2016MobilityAware} have recommended a mobility-aware uplink interference model for 5G networks with Binomial node distribution. However, only the interference from macro users to the small cell users is considered due to their higher power levels. Such an assumption may not be appropriate in a dense environment where the interference levels from individual interferers become less distinguishable as they are located in close proximity. In fact, in dense networks, the variance of interference levels from individual users decays with the node density~\cite{Junyu2017Dense}.
In addition, mmWave specifications and hence the sensitivity of the beams to the blockages in the environment is not taken into account in~\cite{Dong2016MobilityAware}. There have been several prior works that model the effect of blockages~\cite{Bai2014Blockage,gupta2017macro,muller2017analyzing,Venugopal2016blockagefinite,thornburg2016performance}. However, in order to make the calculations tractable, \cite{Bai2014Blockage,gupta2017macro,muller2017analyzing,Venugopal2016blockagefinite,thornburg2016performance} assume that the presence of one obstacle in the path between the transmitters and the receivers completely blocks the line-of-sight (LoS) link. Such an assumption may be justifiable in case of relatively long-distance links. However, in many practical applications such as indoor environments, outdoor small cells where coverage range is limited or even cases where terminals are equipped with larger number of antennas with wider beamwidths, more than one obstacle is needed to impact the power level, causing link blockage. Authors in~\cite{Margarita2016analysis} have suggested a blockage model, built with tools from stochastic geometry and renewal processes, for mmWave cellular communications considering the receiver dimension. However, the model in~\cite{Margarita2016analysis} is based on an unrealistic assumption of very large-sized receiver (i.e., receiver dimension~$ \to \infty $).

\subsection{Contributions}
In this paper, we propose a spatial-spectral interference model for dense finite area 5G mmWave networks while considering realistic blockage effects. The important features of the proposed model can be summarized as follows:

\begin{itemize}
\item In order to capture the effect of blockages properly, we first propose a blockage model that calculates the effective number of the obstacles that results in complete link blockage. Then, we calculate the probability of the presence of that number of obstacles in the path from the interfering APs to the receiver which is the probability of the complete blockage of the interference link. Therefore, unlike prior efforts that consider \emph{binary blockage effect}, here we consider the partial blockage effect of every single obstacle that intersects (partially or completely, depending on the blockage size) the signal beamwidth.

\item Binary effect of the blockages is the resultant of the assumption of $0^{\circ}$ signal beamwidth which is not realistic. In fact, considering such an assumption, a single obstacle that occurs within a given area (depends on the size of the obstacles) close to the LoS link (assumed to be with $0^{\circ}$ beamwidth) causes complete link blockage. However, considering non-zero signal beamwidth, there might be cases where a single obstacle creates partial signal blockage. Therefore, we consider non-zero signal beamwidth in order to model the effect of blockages and the interference behavior.

\item Using the proposed blockage model, we derive the moment generating function (MGF) of the aggregated interference power in a finite-sized mmWave network while considering the configuration uncertainty of the nodes in the spectral domain as well as in spatial domain. Unlike prior works that consider configuration randomness only in spatial domain, we derive the spectral distance distribution and include the configuration randomness in the spectral domain, as well.
\end{itemize}
Subsequently, using the proposed interference model, we evaluate the performance and reliability of the desired communication link based on the average bit error rate (BER) and signal outage probability metrics. It is important to note that average BER and the outage probabilities are the two metrics that can be used in order to evaluate the network performance in fast and slow varying spatial-spectral configuration scenario, respectively. In fact, when the spatial-spectral configuration changes rapidly, it is meaningful to average the interference analysis over all possible realization of the spatial-spectral configurations motivating the use of the average BER metric. However, in the slow varying scenario, the configuration changes slowly and it is more reasonable to calculate the interference for the given configuration and utilize the outage probability metric.

\subsection{Organization}
The remainder of the paper is organized as follows. Section~\ref{sec:sys_model} and~\ref{sec:block_mdl} describe the system model and the distribution of the number of APs considering the effect of the blockages, respectively. In Section~\ref{sec:Intf_statistics}, we calculate the interference statistics and quantify the network performance based on BER and signal outage probability metrics. Finally, Section~\ref{sec:Simulation} highlights the numerical results and its validation using Monte-Carlo based system simulation.
\
\subsection{Notations}
${\mathbb{E}} [.]$ and $\text{Pr} (.)$ denote the expected value and probability measure of the argument, respectively. $F_X (.)$, $M_X (.)$ and $G_X (.)$ are used for representing the cumulative density, moment and probability generating functions of random variable $x$, respectively. $\text{erf}(.)$ represents the error function. $\left| . \right|$ and $\left\| . \right\|$ are the absolute value and $\ell^2$-norm operators, respectively. ${}_1{F_1}\left( {a;b;c} \right)$ denotes the confluent hypergeometric function of the first kind. $\Gamma (.)$ represents the gamma function. Finally, $\text{min}(.)$ and $\text{max}(.)$ are used to denote the minimum and maximum values of the argument, respectively.
\vspace{0.5cm}
\section{System Model} \label{sec:sys_model}
\begin{table}[t]
\caption {\textsc{Summary of System Model Notations}} \label{tab:Notation}
\centering
\begin{tabular}{c|p{6.5cm}}
\hline
\hline
$\textbf{Notation}$         & \textbf{Description}   \\ \hline
$R$         & Radius of the area   \\ \hline
$N$         & Total number of interfering APs in the area of interest  \\ \hline
$p$         & Success probability of the BPP model (model of interfering APs locations)  \\ \hline
$\rho$      & Parameter of the PPP model (model of blockages locations)  \\ \hline
$v_0$ (resp. $f_0$)        & Location (resp. frequency) of the reference receiver   \\ \hline
$v_i$ (resp. $f_i$)        & Location (resp. frequency) of the $\text{i}^{th}$ interfering AP   \\ \hline
$f_s$ (resp. $f_e$)        & Minimum (resp. maximum) of the operational bandwidth  \\ \hline
$d_s$ (resp. $d_e$)        & Minimum (resp. maximum) of the radius of the blockages (modeled as circles)  \\ \hline
$W$                        & Bandwidth of the desired signal  \\ \hline
$2\theta$                   & Signal beamwidth   \\ \hline
\hline
\end{tabular}
\end{table}
Fig.~\ref{fig:blockage area} represents the system model of interest in the present work. Here, we consider a reference pair of transmitter-receiver communicating over a desired communication link in the presence of $N$ number of interfering APs in a circular area with radius $R$ and frequency range $[f_s,f_e]$. Interfering APs are distributed based on a BPP~\footnote{A popular model for finite-sized networks with a given number of nodes~\cite{haenggi2012stochastic}.} in the space-frequency domain with success probability $p$. In other words, we consider a grid structure where the total  $N$ interfering APs are randomly located at space-frequency locations based on Binomial point process~\footnote{Reference transmitter-receiver pair is not a part of the point process.}. The overall received interference signal is the sum of the received signal from each element at random space-frequency location. We also assume that the receiver is at an arbitrary location ${v_0} \in B(O;R) = \left\{ {\left. x \in {\rm I\!R}^2 \right|\,\,{\left\| x \right\|_2} < R} \right\}$ transmitting signals with an arbitrary frequency ${f_0} \in [f_s,f_e]$.
There are also random number of random-sized blockages in the environment. Similar to~\cite{venugopal2016millimeter,thornburg2016performance,muller2017analyzing,Bai2014Blockage}, we assume that blockages are PPP distributed with parameter $\rho$. For a quick reference, we provide the system model parameters in Table~\ref{tab:Notation}. Due to the presence of the arbitrary blockages in the environment, the transmitted signal of interfering APs may be blocked and not all of the interfering APs contribute to the total received interference signal. Therefore, we are primarily concerned with the interferers that are in the LoS link of the reference receiver. In Section \ref{sec:block_mdl}, we provide the proposed blockage model and calculate the probability of each interfering APs being blocked. Subsequently, using the derived probability of the blockage, we obtain the distribution of the number of active (non-blocked) interfering APs.

\section{Blockage Model} \label{sec:block_mdl}
In order to calculate the probability of ${i^{{\rm{th}}}} \in \left\{ {1,2,...,N} \right\}$ interfering AP being blocked, as shown in Fig.~\ref{fig:blockage area}, we consider a radiation cone, denoted by $C_i$, where the edges are determined by the beamwidth of the antenna ($2\theta$). There are random number of blockages, modeled as circles with uniformly distributed radius $d$ in $[d_s,d_e]$, within the path from the interfering APs to the reference receiver. It is important to note that blockages can be at any distance from interfering APs. Therefore, we assume that $r$ is uniformly distributed in $[0, \ell_i]$. Here, $\ell_i$ is a random variable that represents the distances from the $ i^ {\rm{th}}$ interfering AP to the reference receiver. Given the BPP assumption for the locations of the interfering APs, the distribution of $\ell_i$ is given
by\footnote{We drop the subscript $i$ for notational simplicity.}~\cite{Binomial2017Afshang}
\begin{align} \label{eq:distance_distribution}
{f_L}\left( \ell  \right) {=} \left\{ \begin{array}{l}
\frac{{2\ell }}{{{R^2}}}\,\,\,\,\,\,\,\,\,\,\,\,\,\,\,\,\,\,\,\,\,\,\,\,\,\,\,\,\,\,\,\,\,\,\,\,\,\,\,\,\,\,\,\,\,\,\,\,\,\,\,\,\,\,\,\,\,\,\,\,\,\,\,0 < \ell  \le R - \left\| {{v_0}} \right\|\\
\frac{{2\ell {{\cos }^{ - 1}}\left( {\frac{{{{\left\| {{v_0}} \right\|}^2} - {R^2} +{\ell ^2}}}{{2\ell \left\| {{v_0}} \right\|}}} \right)}}{{\pi {R^2}}}\,\,\,\,R - \left\| {{v_0}} \right\| < \ell  \le R + \left\| {{v_0}} \right\|
\end{array} \right.
\end{align}
To calculate the blockage probability, we divide the distance $r$ (distance from the interfering AP to the obstacles in its corresponding radiation cone) into two intervals, (1)~$r \le \frac{d}{{\tan (\theta )}}$ and (2)~$r >  \frac{d}{{\tan (\theta )}}$. In the former interval, only one obstacle blocks the radiation cone, while in the latter more than one blockage is needed to lose the LoS link from the interfering AP to the reference receiver. We calculate the blockage probabilities in both cases (denoted as ${{p_{{\rm{b}}1}}}$ and ${{p_{{\rm{b}}2}}}$, respectively) and average over all realizations of $r$.
Since in the first interval, only one obstacle blocks the entire radiation cone, the probability of blockage is the probability of at least one Poisson-distributed blockage located in the upper triangle in Fig.~\ref{fig:Shadow} which is given by
\begin{align}\label{eq:pb1_cond}
\left. {{p_{{\rm{b}}1}}} \right|d = 1 - {{\rm{e}}^{ - \rho \frac{{{d^2}}}{{\tan (\theta )}}}},
\end{align}
and given the uniform distribution of $d$, we have
\begin{align}\label{eq:pb1} \notag
{p_{{\rm{b}}1}}&= \int\limits_{{d_s}}^{{d_e}} {\Big(1 - {{\rm{e}}^{ - \rho \frac{{{d^2}}}{{\tan (\theta )}}}}\Big)\frac{1}{{{d_e} - {d_s}}}{\rm{d}}} d\\ \notag
&=1 - \frac{{\sqrt {\frac{{\pi \tan (\theta )}}{\rho }} }}{{2\left( {{d_e} - {d_s}} \right)}}\Bigg( {\rm{erf}}\Big( {{d_e}\sqrt {\frac{\rho }{{\tan (\theta )}}} } \Big) \\
&\hspace{4.5cm}- {\rm{erf}}\Big( {{d_s}\sqrt {\frac{\rho }{{\tan (\theta )}}} }\Big)  \Bigg).
\end{align}
In order to calculate the blockage probability in the second interval ($r >  \frac{d}{{\tan (\theta )}}$), we borrow the concepts from point process projection along with results from queuing theory. As shown in  Fig.~\ref{fig:Shadow}, by projecting blockages onto the base of the radiation cone, each blockage in the radiation cone causes a shadow (blocked interval) with length $S{=}\frac{2d\ell_i}{r}$ on the base. Based on~\cite{garcia2008spatial}, the point process obtained by the projection of the points of a PPP from a random subset of a higher dimension onto a lower dimension subspace forms a PPP. Therefore, the number of the shadows of the blockages on the base follows a PPP.
As shown in Fig.~\ref{fig:Shadow}, the resulting {blockages\textquotesingle} shadows (gray lines on the base of the radiation cone) may overlap with one another. We consider the overall overlapped shadow until the next upcoming non-blocked interval as a single resultant shadow (black lines on the base of the radiation cone) with length $S_{\text{res}}$. It is worth mentioning that the number of the resultant shadows is a thinned version of the original PPP obtained from the projection of the blockages in the radiation cone onto its base. In order to calculate the the density of the thinned PPP and also the resultant {shadows\textquotesingle} length, $S_{\text{res}}$, we model the overall projection process by an M/G/$\infty $ queuing
system~\footnote{M/G/$\infty$ is a queuing system with Poisson arrival of customers, infinite servers and general service time distribution.}, in which the initiation and the length of the shadows corresponds to the customer arrival (with poisson distribution) and their service times in the queue system, respectively. In fact, upon mapping the base of the radiation cone to the time duration $[0, 2\ell_i\rm{tan}(\theta)]$, we can model the initiation of the {blockages\textquotesingle} shadows (customers arrivals) as the Poisson point arrivals in that time duration. Furthermore, they are served immediately upon their arrival (infinite servers in the system) for a time that follows a general distribution (shadow lengths). The assumption of infinite number of servers accounts for the overlapped service times (overlapped shadows). By such correspondence, the interval $[0, 2\ell_i\rm{tan}(\theta)]$ consists of alternate busy (partial blocked segment) and idle (non-blocked segment) periods in the queue system. Therefore, the average length and number of the resultant blocked intervals (busy periods of the queue system), denoted as ${\mathbb{E}}[{S_{\text{res}}}]$ and $N_{{S_{\text{res}}}}$, is obtained via
\begin{align} \label{eq:eff_block_len}
{\mathbb{E}}[{S_{\text{res}}}] = \frac{{{{\rm{e}}^{\rho {\mathbb{E}}\left[ S \right]}} - 1}}{\rho },
\end{align}
and
\begin{align} \label{eq:eff_block_num}
{e^{-\rho {\mathbb{E}}\left[ S \right]}}\left( {1+ 2\rho {\ell _i}\tan \left( \theta  \right)} \right) \le {N_{{S_{\text{res}}}}}\le 1 + 2\rho {\ell _i}\tan \left( \theta  \right),
\end{align}
respectively~\cite{filipe2015infinite}, where ${\mathbb{E}}\left[ S \right]$ is given by
\begin{align} \label{eq:Ave_S} \notag
{\mathbb{E}}\left[ S \right] &= {\mathbb{E}}\left[ {\frac{{2d\ell }}{r}\left| {d,r,\ell } \right.} \right] \\
&= \int\limits_{{d_s}}^{{d_e}} {\int\limits_{\frac{d}{{\tan (\theta )}}}^{R + \left\| {{v_0}} \right\|} {\int\limits_{\frac{d}{{\tan (\theta )}}}^\ell  {\frac{{2d\ell }}{r}{\rm{ }}{f_D}\left( d \right)f\left( {r,\ell } \right){\rm{d}}d\,{\rm{d}}r\,{\rm{d}}\ell } } }.
\end{align}
Here,
\begin{align} \label{eq:r_dist}
f\left( {r,\ell } \right)&= {f_L}\left( \ell  \right)f\left( {\left. r \right|\ell } \right)= \left\{ \begin{array}{l}
{f_L}\left( \ell  \right)\frac{1}{\ell }\,\,\,\,\,\,\,\,\,\,\,0 \le r \le \ell \\
\,\,\,\,\,0\,\,\,\,\,\,\,\,\,\,\,\,\,\,\,\,\,\,\,\,\,\,\,{\rm{otherwise.}}
\end{array} \right.
\end{align}
Given the distribution of $\ell_i$ in \eqref{eq:distance_distribution} and the fact that both upper and lower bounds of ${N_{{S_{\text{res}}}}}$ are affine functions of $\ell_i$, we can apply Jensen\textquotesingle s inequality. Therefore, the average number of the the resultant shadows can be bounded by
\begin{align} \label{eq:eff_block_num_ave}
{{\rm{e}}^{-\rho {\mathbb{E}}\left[ S \right]}}\left( {1 + 2\rho {{\mathbb{E}}[\ell]}\tan \left( \theta  \right)} \right) \le {N_{{S_{\text{res}}}}} \le 1 + 2\rho {{\mathbb{E}}[\ell]}\tan \left( \theta  \right).
\end{align}
\begin{figure}[t]
\centering
\includegraphics[scale=0.5]{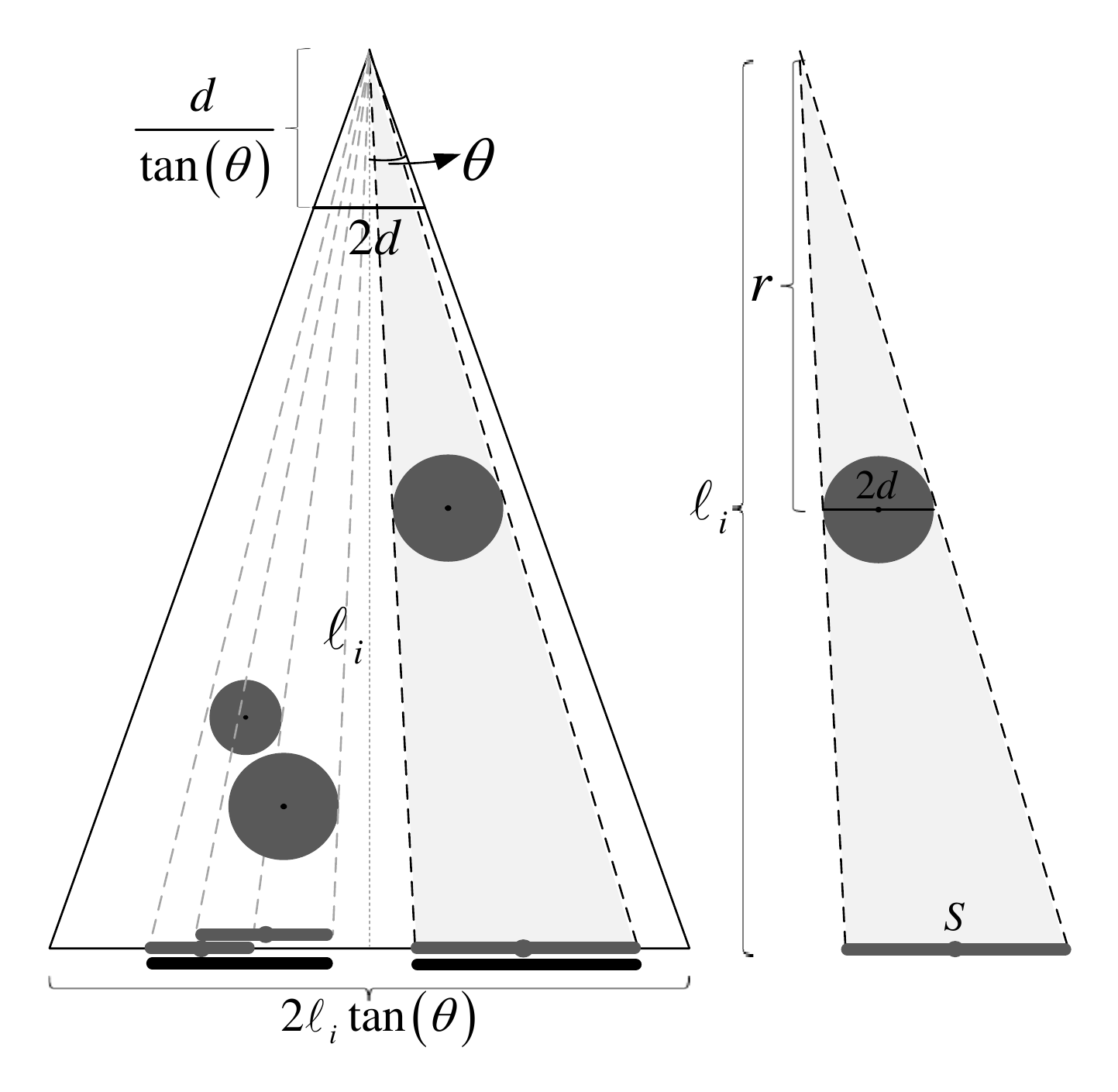}
\caption{Effective shadow of the blockages on the base of the radiation cone.}
\label{fig:Shadow}
\end{figure}

In order to have a complete blockage of the base with average length $2{{\mathbb{E}}[\ell ]}\tan \left( \theta  \right)$, $\Big\lceil {\frac{{2{{\mathbb{E}}[\ell ]}\tan \left( \theta  \right)}}{{\frac{{{{\rm{e}}^{\rho {\mathbb{E}}\left[ S \right]}} - 1}}{\rho }}}} \Big\rceil$ number of resultant shadows with average length ${\mathbb{E}}[{S_{\text{res}}}]$ are needed. It is worth reiterating that the resultant equivalent shadows do not overlap with one another. Therefore, following the fact that resultant shadows are Poisson distributed with density ${N_{{S_{\text{res}}}}}$, the probability of having $\Big\lceil {\frac{{2{{\mathbb{E}}[\ell ]}\tan \left( \theta  \right)}}{{\frac{{{{\rm{e}}^{\rho {\mathbb{E}}\left[ S \right]}} - 1}}{\rho }}}} \Big\rceil$ number of resultant shadows on the base of the radiation cone, i.e., the probability of each interfering AP being completely blocked is given by
\begin{align} \label{eq:pb2}
{p_{{\rm{b}}2}}  = \frac{{{N_{{S_{{\rm{eff}}}}}}^{\left\lceil {\frac{{2{\mathbb{E}}\left[ \ell  \right]\tan \left( \theta  \right)}}{{\frac{{{{\rm{e}}^{\rho {\mathbb{E}}\left[ S \right]}} - 1}}{\rho }}}} \right\rceil }{{\rm{e}}^{ - {N_{{S_{{\rm{eff}}}}}}}}}}{{\left\lceil {\frac{{2 {\mathbb{E}}\left[ \ell  \right]\tan \left( \theta  \right)}}{{\frac{{{{\rm{e}}^{\rho {\mathbb{E}}\left[ S \right]}} - 1}}{\rho }}}} \right\rceil !}}.
\end{align}
Consequently, the overall blockage probability of each interfering AP, in the average sense, is obtained by
\begin{align} \label{eq:blockage_prob} \notag
{p_{\rm{b}}} &= \Pr \left( {r \le \frac{{{\mathbb{E}}\left[ d \right]}}{{\tan (\theta )}}} \right){p_{{\rm{b1}}}} + \Pr \left( {r > \frac{{{\mathbb{E}}\left[ d \right]}}{{\tan (\theta )}}} \right){p_{{\rm{b2}}}}\\
&=\frac{1}{{\frac{{{\mathbb{E}}\left[ d \right]}}{{\tan \left( \theta  \right)}}}}{p_{{\rm{b1}}}} + \frac{1}{{{\mathbb{E}}\left[ \ell  \right] - \frac{{{\mathbb{E}}\left[ d \right]}}{{\tan \left( \theta  \right)}}}}{p_{{\rm{b2}}}}.
\end{align}
Given the upper and lower bounds of the number of resultant shadows in~\eqref{eq:eff_block_num_ave}, the blockage probability in~\eqref{eq:blockage_prob} is shown for different network parameters in Fig.~\ref{fig:pb_rho}, \ref{fig:pb_theta} and \ref{fig:pb_R}. As we can see, the gap between the upper and lower bounds is small. It is worth mentioning that, in this blockage model, we assume the receiver dimension spans the base of the radiation cone. As we can see in Fig.~\ref{fig:pb_rho}, the blockage probability increases with increasing density of the obstacles, as expected. By increasing the beamwidth, in one hand the chance of signals intersecting with more obstacles in the environment increases. On the other hand, since the beam is wider, even after partial blockage due to an obstacle, the signal can be partially received by the receiver. Both effects are captured in the proposed blockage model by considering the blockage density and the beamwidth. However, as shown in Fig.~\ref{fig:pb_theta}, the overall effect is in a way that the probability of blockage decreases as the beam gets wider. Moreover, by increasing the cell radius (which implies an increase in the average distance between the interfering APs and reference receiver) the probability of blockage increases. This is due to the fact that, considering the beamwidth of the signal confined to the receiver\textquotesingle s dimension and increasing the average distance, the signal beamwidth becomes narrower with respect to the {obstacles\textquotesingle} dimension. Therefore, the chance of getting blocked and loosing the LoS link increases, which is consistent with the 3GPP~\cite{GenerationPartnershipProject2010} and potential 5G models~\cite{gupta2017macro,muller2017analyzing,Venugopal2016blockagefinite,thornburg2016performance,Katsuyuki2016GPP} as well, where the probability of having LoS decreases exponentially as the length of the link increases.
\begin{figure}[t]
\centering
\includegraphics[width=8cm,height=6cm]{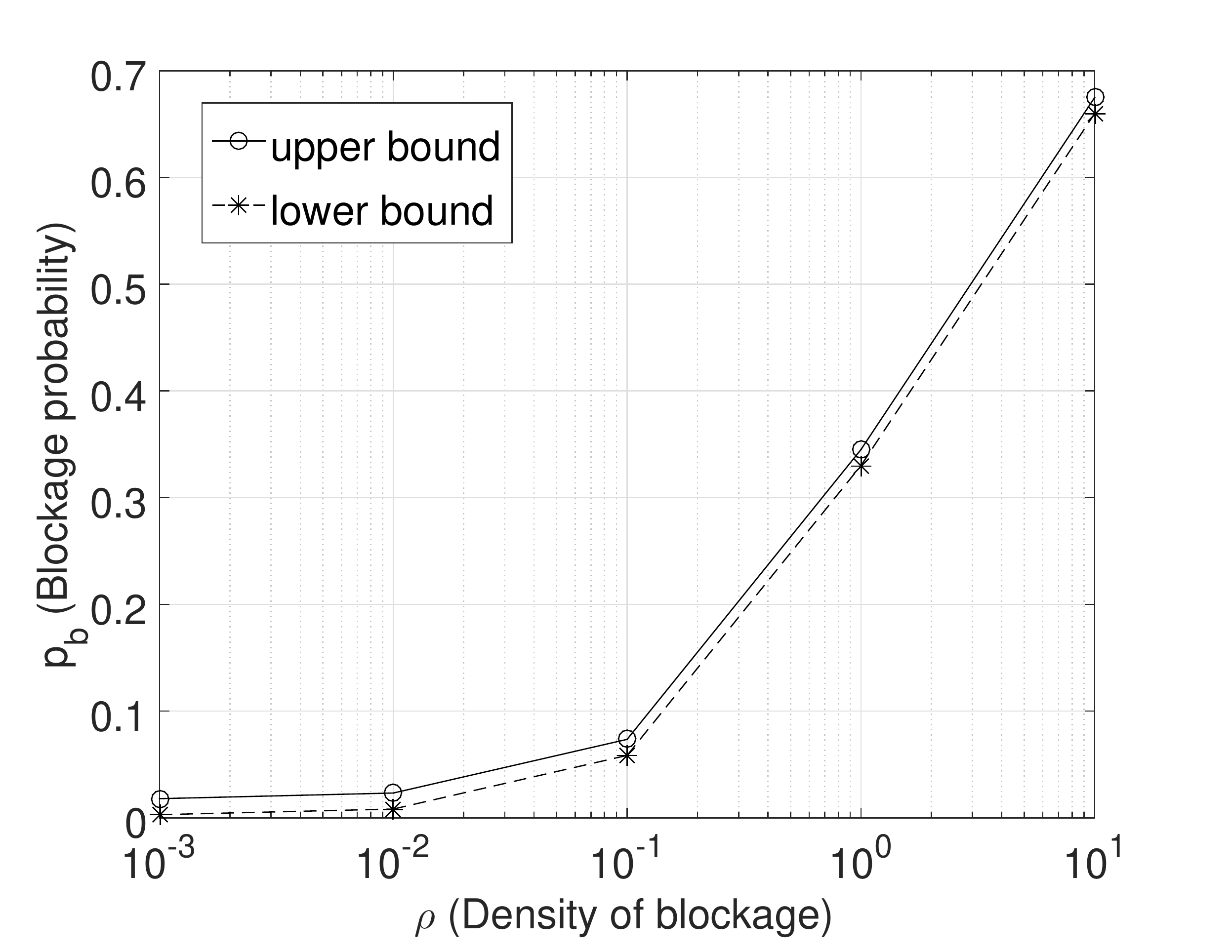}
\caption{Blockage probability vs. $\rho$ values. Here, $R=20$m, $v_0=10$m, $d_s=0.2$m, $d_e=0.8$m and $\theta=20^{\circ}$.}
\label{fig:pb_rho}
\end{figure}
\begin{figure}[t]
\centering
\includegraphics[width=8cm,height=6cm]{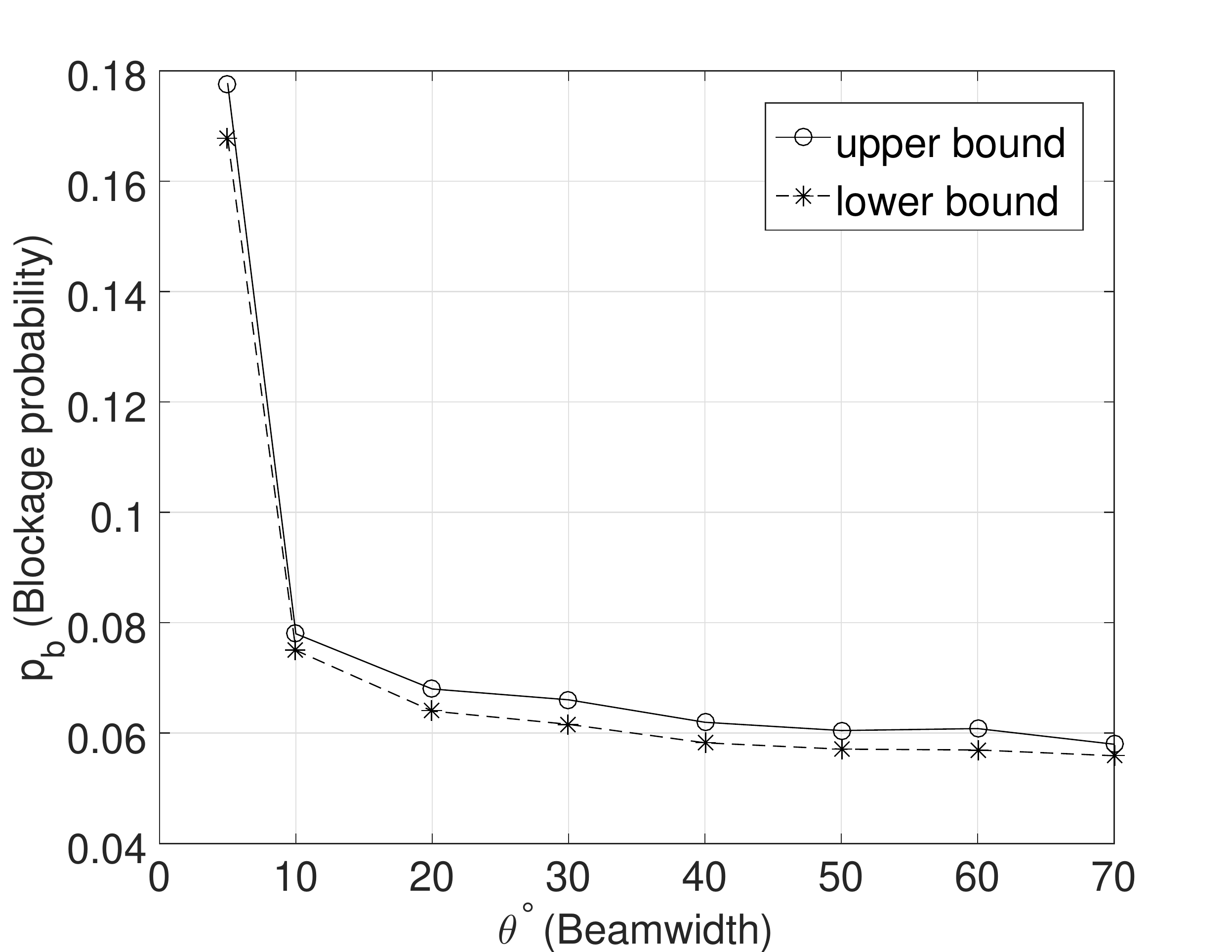}
\caption{Blockage probability vs. $\theta$ values. Here, $R=20$m, $v_0=10$m, $d_s=0.2$m, $d_e=0.8$m and $\rho=10^{-1}$.}
\label{fig:pb_theta}
\end{figure}
\begin{figure}[t]
\centering
\includegraphics[width=8cm,height=6cm]{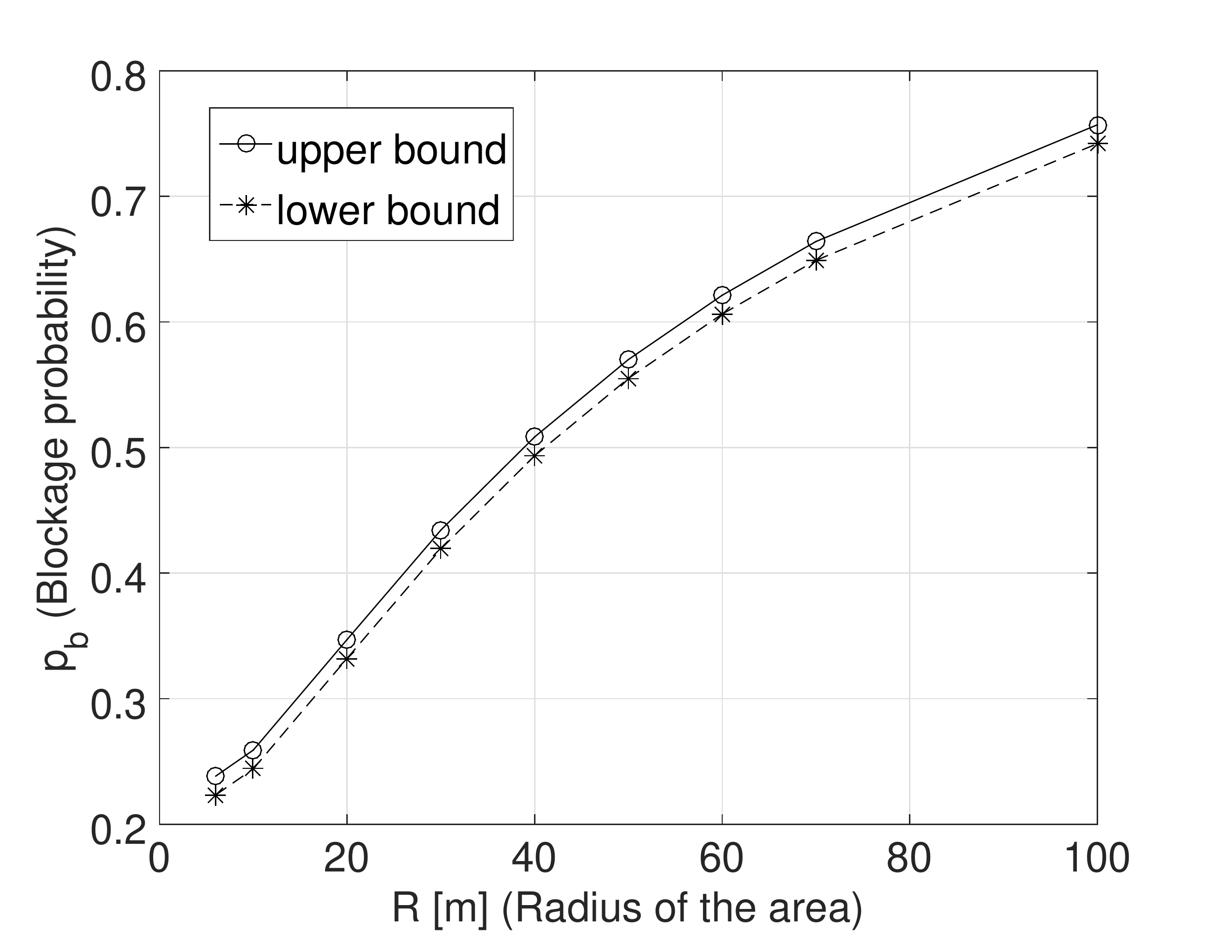}
\caption{Blockage probability vs. $R$ values. Here, $v_0=5$m, $d_s=0.2$m, $d_e=0.8$m and $\rho=10^{-2}$.}
\label{fig:pb_R}
\end{figure}

Considering the blockage probability in \eqref{eq:blockage_prob}, the distribution of the total number of non-blocked interfering APs is calculated using the following lemma:
\newtheorem{lemma}{Lemma}
\begin{lemma}\label{lem:non-blocked}
The total number of non-blocked interfering APs, denoted as $K$, is a Binomial random variable with success probability $p\left( {1 - {p_{\rm{b}}}} \right)$.
\end{lemma}
\begin{proof}
Let $K = {K_1} + {K_2} + ... + {K_N}$, where $K_i$ is a Bernoulli random variable and equals $1$, if the $i^{\rm {th}}$ interfering AP is not blocked, and $0$, otherwise. Therefore, the probability generating function (PGF) of $K_i$ is given by
\begin{equation}\label{eq:PGF}
{G_{K_i}}\left( {\rm{z}} \right) = (1 - {p_{\rm{b}}}){\rm{z}} + {p_{\rm{b}}}.
\end{equation}
Subsequently, we have
\begin{align}\label{eq:PGF_nonBlocked}\notag
{G_K}\left( {\rm{z}} \right)&= {\mathbb{E}}\Big[ {{{\mathop{\rm z}\nolimits} ^{\sum\limits_{i = 1}^N {{K_{i}}} }}} \Big]= \sum\limits_{k \ge 0} {{{\left( {{\mathbb{E}}\left[ {{{\mathop{{\rm z}}\nolimits} ^{K_i}}} \right]} \right)}^k}p\left( {N = k} \right)}\\ \notag
&= {G_N}\left( {{G_{K_i}}({\mathop{\rm z}\nolimits} )} \right)= {\left[ {\left( {1 - p} \right) + p\left( {(1 - {p_{\rm{b}}}){\rm{z}} + {p_{\rm{b}}}} \right)} \right]^N} \\
&= {\left[ {1 - p\left( {1 - {p_{\rm{b}}}} \right) + p\left( {1 - {p_{\rm{b}}}} \right){\rm{z}}} \right]^N}.
\end{align}
 which is the PGF of a Binomial random variable with success probability $p\left( {1 - {p_{\rm{b}}}} \right)$.
\end{proof}

Now, given the distribution of the number of active interfering APs, in Section~\ref{sec:Intf_statistics} we derive the distribution of the aggregated interference power at the reference receiver. Using the derived distribution, we provide expressions for the error performance of the desired communication link in terms of average BER and outage probabilities.

\section{Interference Statistics And System Performance} \label{sec:Intf_statistics}
In this section, the MGF of the received aggregated interference power at the reference receiver is derived considering the configuration randomness of the interfering APs in both spectral and spatial domains. Given the distribution of the active interfering APs, denoted as $K$ in~\eqref{eq:PGF_nonBlocked} the aggregated interference power at an arbitrarily located reference receiver is
\begin{align} \label{eq:received_power}
I_{\rm{agg}}={\sum\limits_{i = 1}^K {{{\cal P}_{{I_i}}}} }
\end{align}
where, $\mathcal{P}_{I_i}$ is the effective received interference power from the $i^\text{th}$ interfering AP at the output of the matched filter which is obtained by~\cite{Hamdi2009unified}
\begin{align} \label{eq:Intf_pow}
{{{\cal P}_{{I_i}}}}={q_i}{h_i}{\left\| {\ell_i} \right\|^{ - \alpha }}\Upsilon \left( {\omega_i} \right).
\end{align}
Here, $h_i$ and ${\left\| . \right\|^{ - \alpha }}$ capture the Nakagami-$m$ small scale fading and pathloss effects, respectively. $\ell_i={v_0} - {v_i}$ and $\omega_i={f_i}-{f_0}$ denote the spatial and spectral distance between the $i^{\rm{th}}$ interfering AP (located at arbitrary spatial-spectral location \{$v_i$,$f_i$\}) and the reference receiver, respectively. $q_i$ is the transmitted power of the $i^{\rm{th}}$ interfering AP.  Moreover, $\Upsilon \left( {{\omega_i}} \right)$ is defined as
\begin{align}
\Upsilon(\omega_i)=\int\limits_{{f_0} - \frac{W}{2}}^{{f_0} + \frac{W}{2}} {\Phi \left( {f - {f_i}} \right){{\left| {H\left( {f - {f_0}} \right)} \right|}^2}{\mathop{\rm d}\nolimits} f} ,
\end{align}
where, $H(f-f_0)$ is the transfer function of the matched filter at the reference receiver with arbitrary center frequency $f_0$ and bandwidth $[-\frac{W}{2},\frac{W}{2}]$, and $\Phi(f-f_i)$ is the power spectral density of the baseband equivalent of the $i^{\rm{th}}$ interfering APs signals with frequency $f_i$. Considering~\eqref{eq:Intf_pow}, as ${\left\| {\ell_i} \right\|^{ - \alpha }}$ captures the impact of spatial distances (and thereby random spatial configuration), $\Upsilon(\omega_i)$ accounts for the effect of frequency separation (and thereby random spectral configuration) in the interference power. The distribution of the aggregated interference power is obtained using the following theorem:
\begin{theorem} \label{theo:theo1}
The MGF of the aggregated interference power at the arbitrarily located receiver, is given by
\begin{align} \label{eq:MGF_alt_theo}
{M_{{I_{\rm agg}}}}({\rm{s}}) ={\left[ {1 - p\left( {1 - {p_{\rm{b}}}} \right) + p\left( {1 - {p_{\rm{b}}}} \right){M_{{{{\cal P}_{{I_i}}}}}\left( \rm{s} \right)}} \right]^N},
\end{align}
where
\begin{align} \label{eq:MGF_indivi}
&{M_{{{\cal P}_{{I_i}}}}}\left( {\rm{s}} \right)\,\, = \sum\limits_{n = 0}^\infty  {\frac{{{(q\, {\rm s})^n}}}{{n!}}\frac{{{m^{ - n}}\Gamma \left( {n + m} \right)}}{{\Gamma \left( m \right)}}\frac{{2\gamma_n \left( {{f_s},{f_e}} \right)\kappa_n \left( {R,{v_0}} \right)}}{{{R^2}\left( {{f_e} - {f_s}} \right)}}}  .\\ \notag
\end{align}
Here, ${M_{{{\cal P}_{{I_i}}}}}\left( {\rm{s}} \right)$ is the MGF of the $i^ {\rm{th}}$ interfering AP power, where
\begin{align}
{\gamma_n \left( {{f_s},{f_e}} \right)}=\hspace{-0.3cm}{\int\limits_0^{\min \left( {\left| {{\omega _e}} \right|,\left| {{\omega _s}} \right|} \right)} \hspace{-0.3cm}{\Upsilon {{\left( \omega  \right)}^n}{\rm{d}}\omega }  +\hspace{-0.3cm} \int\limits_0^{\max \left( {\left| {{\omega _e}} \right|,\left| {{\omega _s}} \right|} \right)} \hspace{-0.3cm}{\Upsilon {{\left( \omega  \right)}^n}{\rm{d}}\omega } },
\end{align}
\begin{align} \notag
{\kappa_n \left( {R,{v_0}} \right)}&=\hspace{-0.1cm}\int\limits_0^{R - \left\| {{v_0}} \right\|} {\hspace{-0.3cm}{\ell ^{ - n\alpha  + 1}}\,{\rm{d}}\ell \,\,} \\
&+\hspace{-0.1cm} \int\limits_{R - \left\| {{v_0}} \right\|}^{R + \left\| {{v_0}} \right\|} {\hspace{-0.2cm}\frac{{{\ell ^{ - n\alpha  + 1}}}}{\pi }{{\cos }^{ - 1}}\left( {\frac{{{{\left\| {{v_0}} \right\|}^2} - {R^2} + {\ell ^2}}}{{2\ell \left\| {{v_0}} \right\|}}} \right)\,{\rm{d}}\ell } .
\end{align}
and $\omega_{e}={f_e} - {f_0}$, $\omega_{s}={f_s} - {f_0}$.
\end{theorem}
\begin{proof}
In order to calculate the MGF of the received aggregated interference power, we have
\begin{align} \label{eq:MGF_alt_det} \notag
{M_{{I_{\rm{agg}}}}}({\rm{s}}) &={\mathbb{E}}\left[ {{{\rm{e}}^{{\rm{s}}\sum\limits_{i = 1}^K {{{\cal P}_{{I_i}}}} }}} \right] = \sum\limits_{k \ge 0} {{{\left( {{\mathbb{E}}\left[ {{{\rm{e}}^{{\rm{s}}{{\cal P}_{{I_i}}}}}} \right]} \right)}^k}p\left( {K = k} \right)}   \\ \notag
&= {G_K}\left( {M_{{{{\cal P}_{{I_i}}}}}\left( \rm{s} \right)} \right)\\
&= {\left[ {1 - p\left( {1 - {p_{\rm{b}}}} \right) + p\left( {1 - {p_{\rm{b}}}} \right){M_{{{{\cal P}_{{I_i}}}}}\left( \rm{s} \right)}} \right]^N},
\end{align}
where, ${M_{{{\cal P}_{{I_i}}}}}\left( {\rm{s}} \right)$ is the MGF of the $i^ {\rm{th}}$ interfering AP power and calculated by~\footnote{We drop the subscript $i$ for notational simplicity.}
\begin{align} \label{eq:MGF_indivi_det} \notag
{M_{{{\cal P}_{{I_i}}}}}\left( {\rm{s}} \right) &= {\mathbb{E}}\left[ {{{\rm{e}}^{{\rm{s}}\,q\,h\,{{ \ell  }^{ - \alpha }}\Upsilon \left( \omega  \right)}}} \right] \\ \notag
&= \int\limits_0^\infty  \int\limits_0^{R + \left\| {{v_0}} \right\|} \int\limits_0^{\max \left( {\left| {{\omega _e}} \right|,\left| {{\omega _s}} \right|} \right)} {{\rm{e}}^{{\rm{s}}\,q\,h\,{{ \ell  }^{ - \alpha }}\Upsilon \left( \omega  \right)}}\\
&\hspace{3 cm} \times \,\,{f_\Omega }\left( \omega  \right){f_L}\left( \ell  \right)f\left( h \right){\rm{d}}\omega \,{\rm{d}}\ell \,{\rm{d}}h.
\end{align}
Given the BPP assumption of the location of interferer in the space-frequency domain, the distributions of spectral distance~\footnote{Detailed derivation of the distribution can be found in the Appendix.} is derived as
\begin{align}\label{eq:frequency_distribution}
&{f_\Omega }\left( \omega  \right) {=} \left\{ \begin{array}{l}
\frac{2}{{{f_e} - {f_s}}}\,\,\,\,\,\,\,\,\,\,\,\,\,\,0 < \omega  \le \min \left( {\left| {\omega_{e}} \right|,\left| {\omega_{s}} \right|} \right)\\
\frac{1}{{{f_e} - {f_s}}}\,\,\,\,\min \left( {\left| {\omega_{e}} \right|,\left| {\omega_{s}} \right|} \right) < \omega  \le \max \left( {\left| {\omega_{s}} \right|,\left| {\omega_{s}} \right|} \right)
\end{array} \right.
\end{align}
where, $\omega_{e}={f_e} - {f_0}$ and $\omega_{s}={f_s} - {f_0}$. Having the spatial and spectral distance distributions given in~\eqref{eq:distance_distribution} and~\eqref{eq:frequency_distribution}, respectively, Nakagami-$m$ assumption of the small scale fading component, i.e., $h$, and using the polynomial expansion of the exponential function, the integral in~\eqref{eq:MGF_indivi_det} reduces to that in~\eqref{eq:MGF_indivi}. Subsequently, by substituting~\eqref{eq:MGF_indivi} in~\eqref{eq:MGF_alt_det} the result in Theorem~\ref{theo:theo1} is obtained.
\end{proof}
In order to evaluate the system performance using the derived interference model in Theorem~\ref{theo:theo1}, we invoke the result in~\cite{Niknam2017Spatial-Spectral}, in which using the MGF of the aggregated interference power, the average BER is calculated by the following expression,
\begin{align} \label{Average_BER} \notag
{\rm{BER}}_{\rm{ave}}&= \frac{1}{2} - \frac{{\sqrt c }}{\pi }\frac{{\Gamma (m + \frac{1}{2})}}{{\Gamma (m)}}\int_0^\infty  {\frac{{{{\kern 1pt} _1}{F_1}( m+\frac{1}{2};\frac{3}{2};-c{\rm{s}})}}{{\sqrt {\rm{s}} }}} \\
 &\hspace{1.25cm}\times {M_{I_{\rm{agg}}}}\left( {-\frac{m}{{{q_0}\ell _0^{ - \alpha }}}{\rm{s}}} \right){{\rm{e}}^{-{\frac{{ m\sigma _n^2}}{{{q_0}\ell _0^{ - \alpha }}} } {\rm{s}}}}{\rm{ds}},
\end{align}
where, $q_0$ and $\ell_0$ denote the power of the reference transmitter and the distance between the reference transmitter-receiver pair. Moreover, $m$ is the shape factor of Nakagami distributed channel and $c$ is a constant that depends on the modulation type. In addition, $\sigma _n^2$ represents the power of the additive noise bandlimited to the signal bandwidth $[-\frac{W}{2},\frac{W}{2}]$.
\begin{figure}[t]
\centering
\includegraphics[width=8cm,height=6cm]{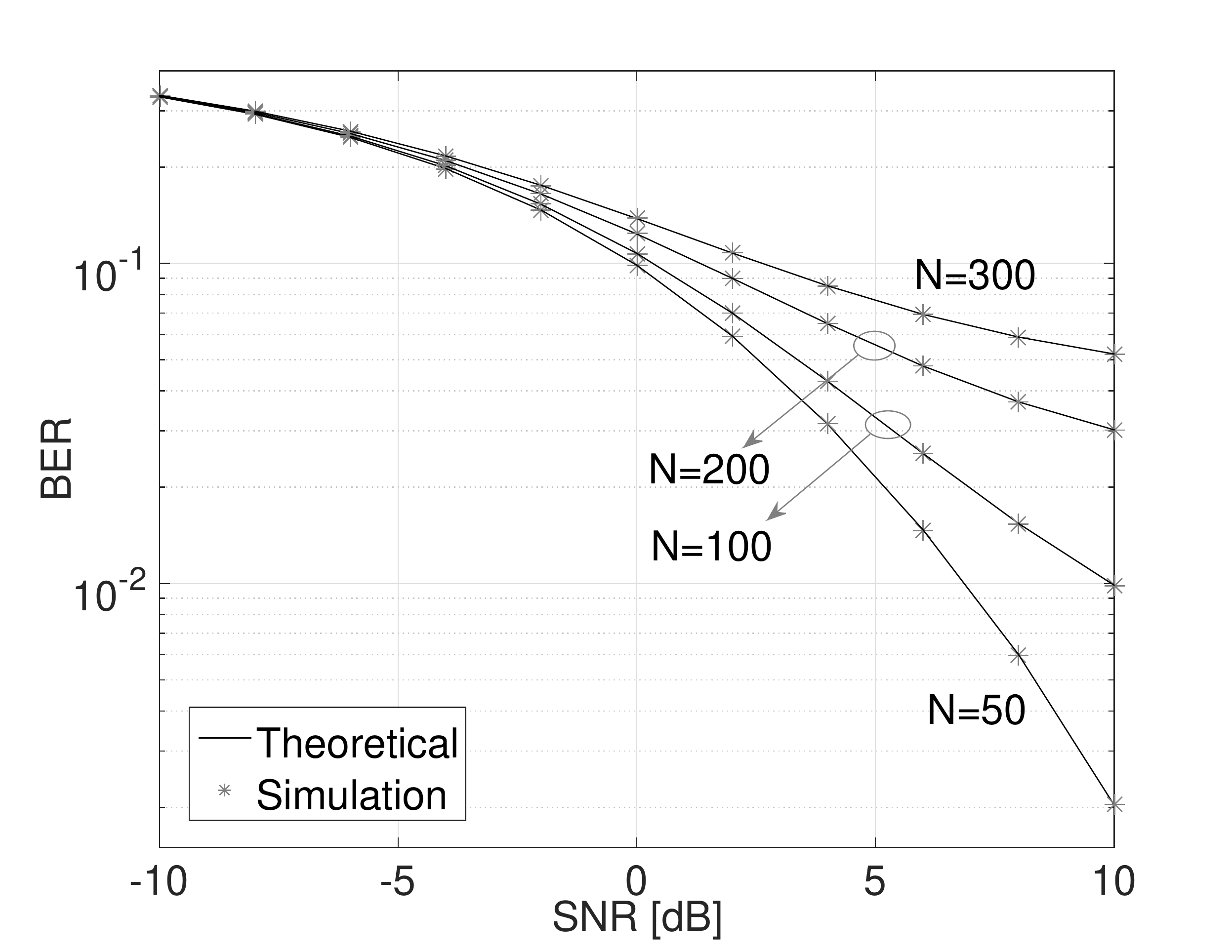}
\caption{Bit error rate versus SNR for different $N$ values, $\rho{=}10^{-2}$.}
\label{fig:BERvsSNR_N}
\end{figure}
\begin{figure}[t]
\centering
\includegraphics[width=8cm,height=6cm]{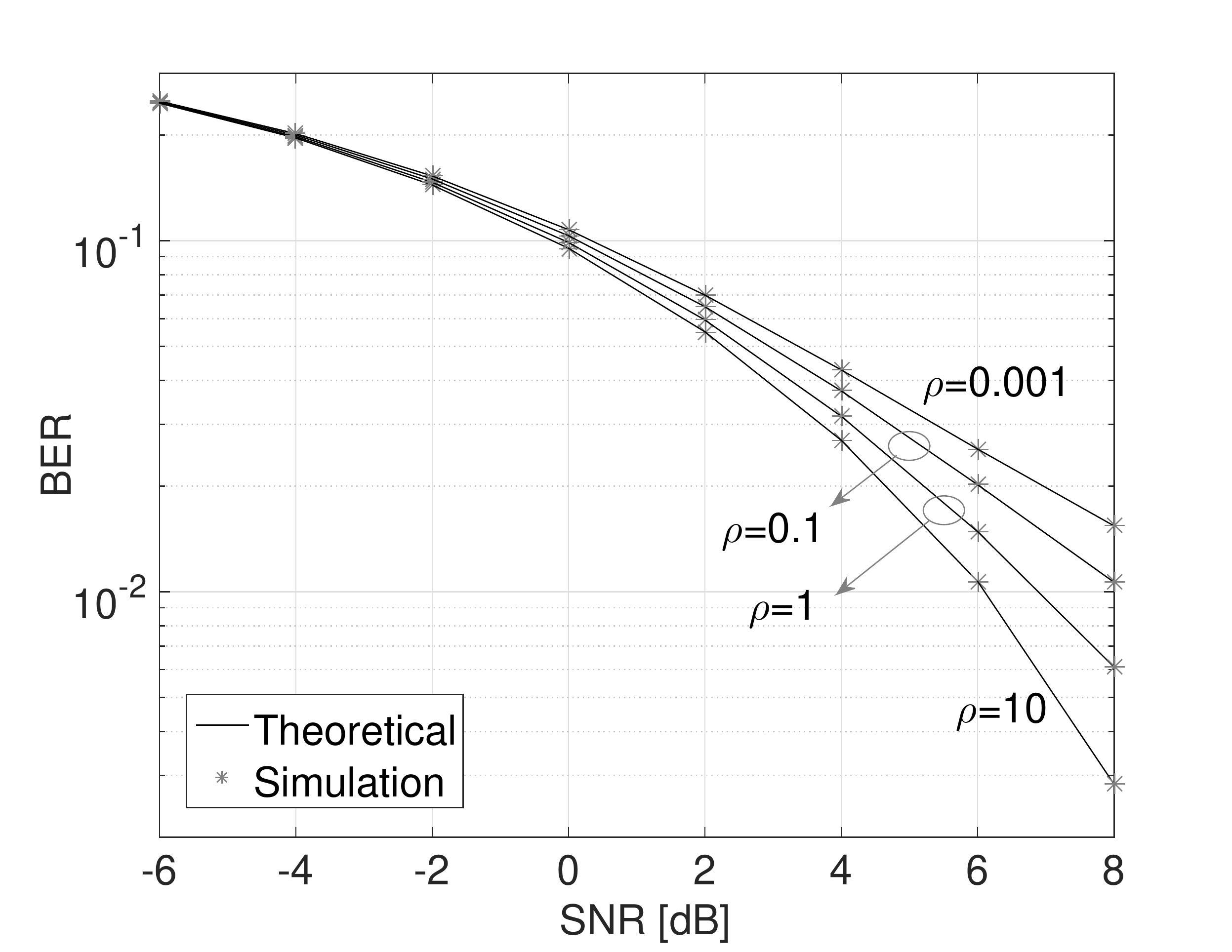}
\caption{Bit error rate versus SNR for different $\rho$ values, $N{=}100$.}
\label{fig:BERvsSNR_rho}
\end{figure}
\begin{figure}[t]
\centering
\includegraphics[width=8cm,height=6cm]{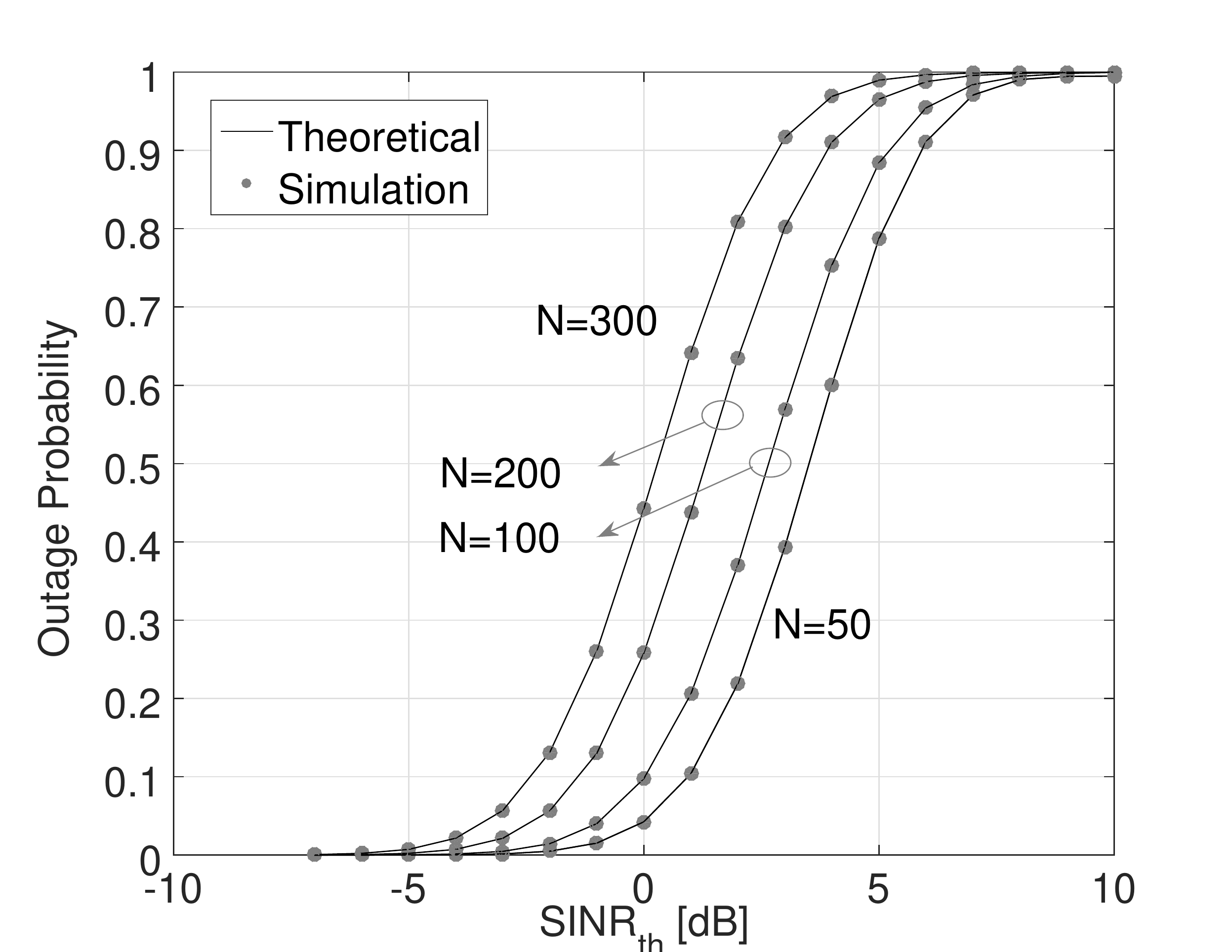}
\caption{Outage probability versus SINR threshold for different $N$ values, $\rho{=}10^{-2}$.}
\label{fig:OutagevsSINR_N}
\vspace{-0.14cm}
\end{figure}
\begin{figure}[t]
\centering
\includegraphics[width=8cm,height=6cm]{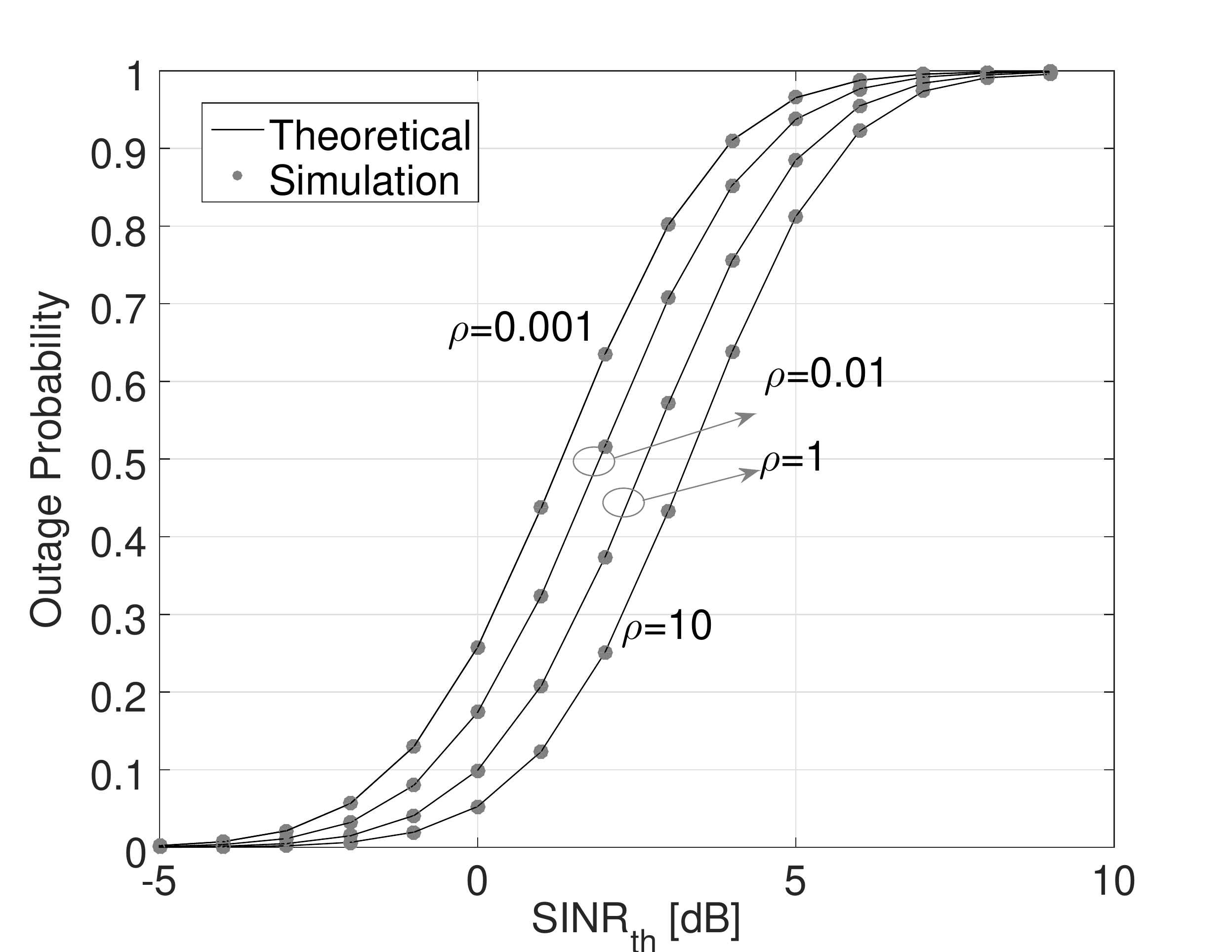}
\caption{Outage probability versus SINR threshold for different $\rho$ values, $N{=}150$.}
\label{fig:OutagevsSINR_rho}
\end{figure}

The outage probability $P_{outage}(\eta)$ is defined as the cumulative distribution function (CDF) of the signal-to-noise-plus-interference ratio (SINR) evaluated at a threshold $\eta$ which is
\begin{align} \label{eq:outage_general}
{P_{outage}} = \text{Pr}\left( {\text{SINR} \le \eta } \right) = \text{Pr}\Bigg( {\frac{{{q_0}{h_0}\ell _0^{ - \alpha }}}{{\sum\limits_{i = 1}^K {{{\cal P}_{{I_i}}}}  + \sigma _n^2}} \le \eta } \Bigg).
\end{align}
In order to calculate the above probability we first rearrange the argument in~\eqref{eq:outage_general} and obtain the following expression
\begin{align} \label{eq:outage_general_rearrange}
{P_{outage}} = \text{Pr}\Bigg( {\overbrace {-\frac{{\eta \sum\limits_{i = 1}^U {{{\cal P}_{{I_i}}}} }}{{\sigma _n^2}} + \frac{{{q_0}{h_0}\ell _0^{ - \alpha }}}{{\sigma _n^2}}}^\Lambda   \le    \eta } \Bigg) = {F_\Lambda }\left( {  \eta } \right),
\end{align}
where, $F_\Lambda$ denotes the CDF of random variable $\Lambda$. Now, using Gil-Pelaez inversion formula
\begin{align} \label{eq:Gil-Pelaez}
{F_\Lambda }\left( \lambda  \right) = \frac{1}{2} - \frac{1}{\pi }\int_0^\infty  {{\mathop{\rm Im}\nolimits} \left\{ {{M_\Lambda }\left( {j{\rm s}} \right){e^{ - j{\rm s}\lambda }}} \right\}\frac{{{\rm{d}}{\rm s}}}{{\rm s}}},
\end{align}
and making the following substitution
\begin{align} \label{eq:multiplication_of_MGF}
{M_\Lambda }\left( {\rm s} \right) = {\left( {1 - \frac{{{q_0}\ell _0^{ - \alpha }{\rm s}}}{{m\sigma _n^2}}} \right)^{ - m}}{M_{I_{\rm{agg}}}}\left( {-\frac{\eta }{{\sigma _n^2}}{\rm s}} \right),
\end{align}
the outage probability can be simplified and written as follows
\begin{align} \label{eq:Outage} \notag
&{P_{outage}}= \frac{1}{2} - \frac{1}{\pi }\int_0^\infty  {\mathop{\rm Im}\nolimits} \Bigg\{ {{\left( {1 - \frac{{j{q_0}\ell _0^{ - \alpha }s}}{{m\sigma _n^2}}} \right)}^{ - m}}\\
&\hspace{4.5cm}\times {M_{I_{\rm{agg}}}}\left( -{\frac{{j\eta }}{{\sigma _n^2}}{\rm s}} \right){e^{-j{\rm s}\eta }} \Bigg\}\frac{{{\rm{d}}{\rm s}}}{{\rm s}}.
\end{align}

\section{Numerical Results} \label{sec:Simulation}
In this section, we present numerical results to characterize the spatial-spectral interference model as a function of network parameters. A circular area of radius $R=25$ is considered. The reference receiver is located at $\left\| {{v_0}} \right\|{=10}$ and $f_0=62$ GHz. Moreover, $f_s$ and $f_s$ are set to 58 GHz and 64 GHz, respectively. We assume the pathloss exponent, $\alpha$, and the shape factor of Nakagami distribution, $m$, are set to 2.5 and 5, respectively. Here, the transmitted power of all interfering APs are assumed to be the same and set to 30 dBm. The beamwidth of the mmWave signals, i.e., $2\theta$, is set to 20 degrees. We assume Gaussian PSD for interfering APs (it can be any PSD shape) and a raised-cosine (RC) matched filter at the reference receiver side. It is worth mentioning that the proposed model is not limited to these assumptions on specific power spectral densities of the desired and interferers' signals.

In Fig.~\ref{fig:BERvsSNR_N} and~\ref{fig:BERvsSNR_rho} , BER versus SNR is shown for different $N$ and $\rho$ values, respectively. As expected, the performance of the system degrades as $N$ increases. The same trend can be observed in Fig.~\ref{fig:BERvsSNR_rho} where by increasing the density of the blockages, larger number of interfering APs are blocked. Therefore, the accumulated interference signal decreases and results in a better performance. This is an important result that indicates mmWave {signals\textquotesingle} sensitivity to blockages can be advantageous in densely deployed networks, where objects and users that serve as obstacles reduce the level of interference.

In Fig.~\ref{fig:OutagevsSINR_N} and \ref{fig:OutagevsSINR_rho}, the error performance of the desired communication link is shown in terms of outage probability. Here, the performance decreases when there is lower number of active interfering APs (i.e., decreasing the density of blockages, $\rho$, and increasing total number of interfering APs, $N$). It is important to note that in all Fig.~\ref{fig:BERvsSNR_N}, \ref{fig:BERvsSNR_rho}, \ref{fig:OutagevsSINR_N} and \ref{fig:OutagevsSINR_rho}, the simulated average BER and outage probability plots align well with the result from the theoretically derived interference model. Moreover, the performance of the desired communication link with and without consideration of the effect of the blockages sensitivity of the interfering links are illustrated in terms of both metrics, i.e., BER and outage probability, in Fig.~\ref{fig:BlockedvsNon_BER} and \ref{fig:BlockedvsNon_Outage}, respectively. As illustrated, unlike traditional interference model, where the impact of the blockages is not considered, directionality of mmWave signals leads to a noticeably different interference profile which is effectively captured by the proposed model. It is worth reiterating, the proposed model considers the uncertainty of the interfering node configuration in both spatial and spectral domains, simultaneously. \vspace{0.5cm}

\begin{figure}[t]
\centering
\includegraphics[width=8cm,height=6cm]{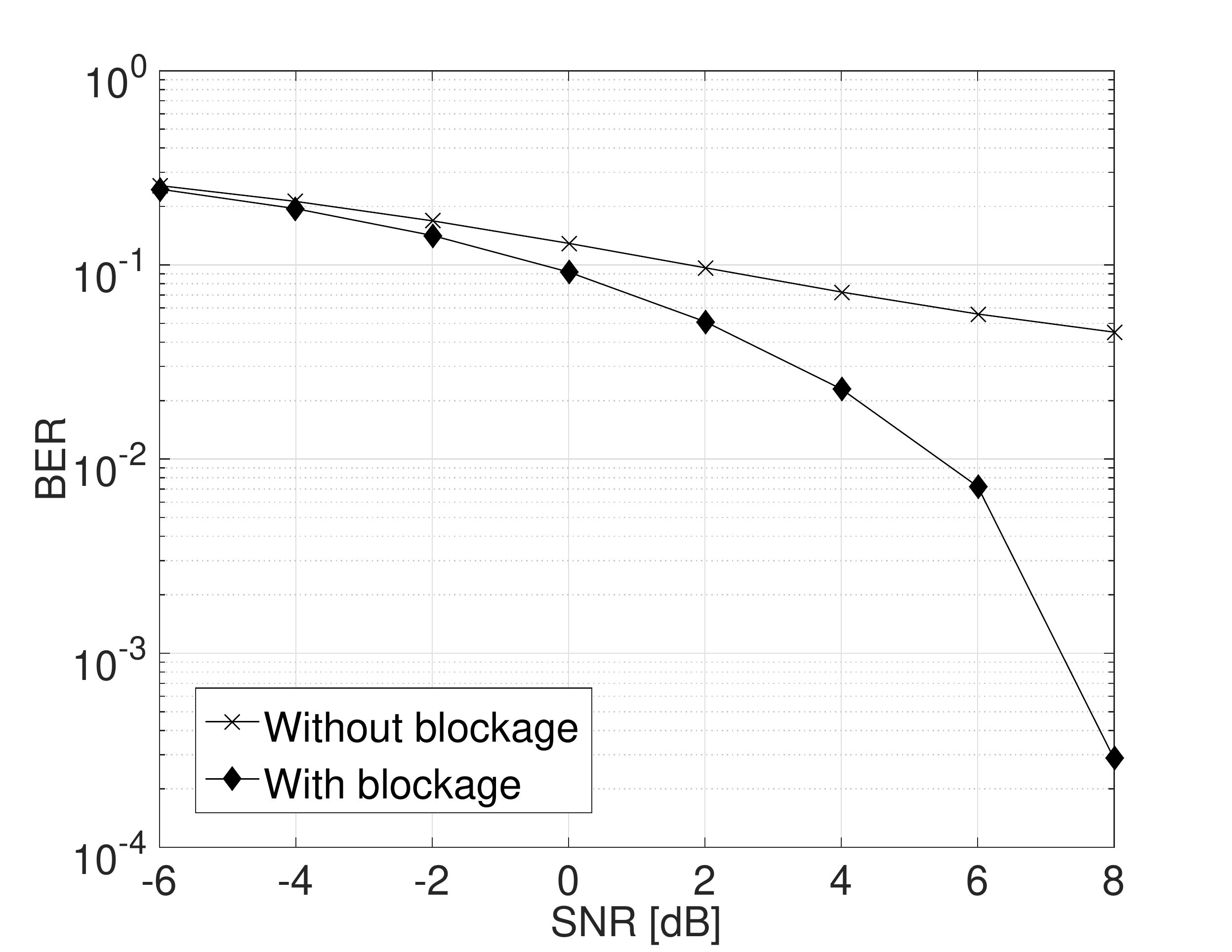}
\caption{Bit error rate versus SNR for $N=200$.}
\label{fig:BlockedvsNon_BER}
\end{figure}
\begin{figure}[t]
\centering
\includegraphics[width=8cm,height=6cm]{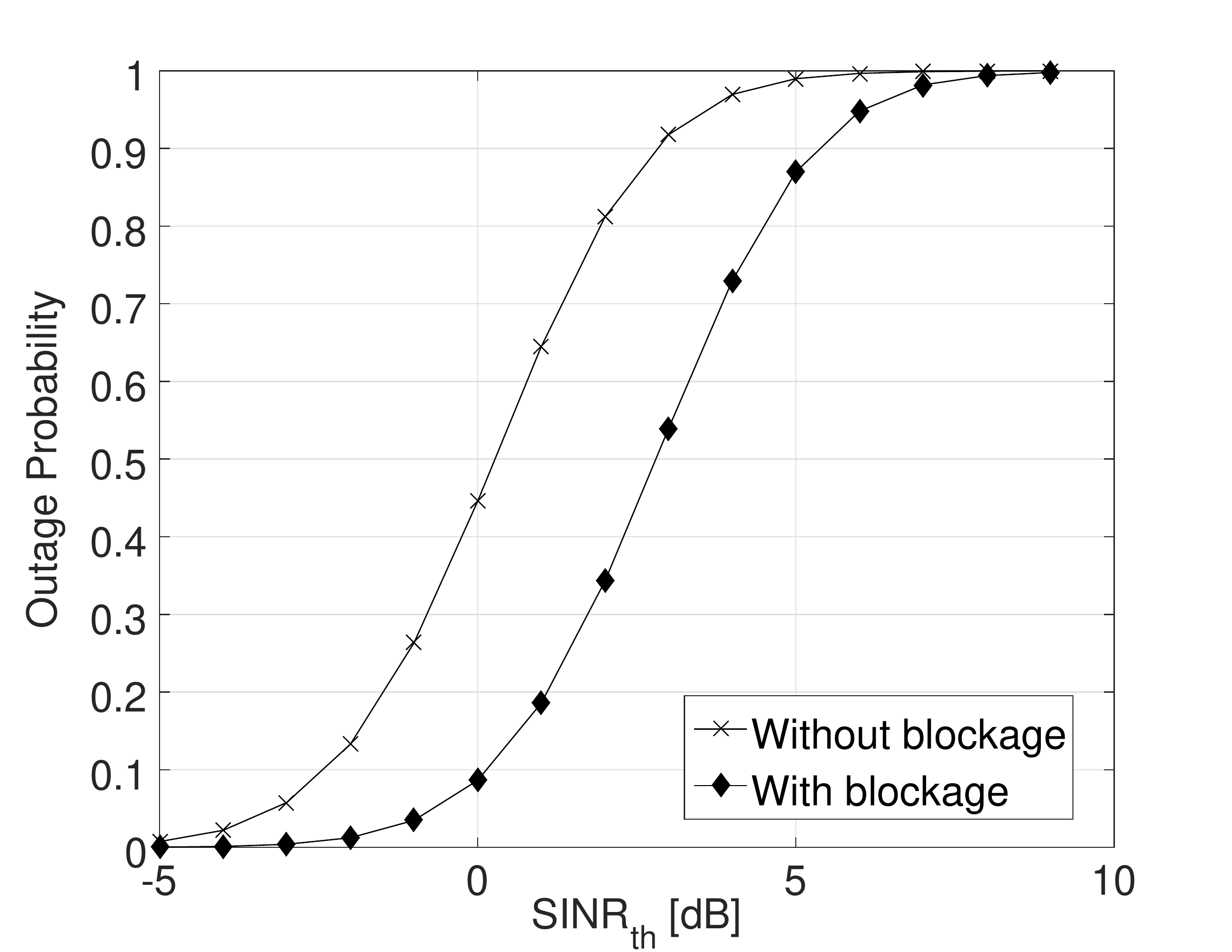}
\caption{Outage probability versus SINR threshold for $N=300$.}
\label{fig:BlockedvsNon_Outage}
\end{figure}

\section{Conclusion} \label{sec:Conclusion}
In this paper, we propose a spatial-spectral interference model for dense finite-area 5G mmWave network considering the effect of blockages on mmWave signals. The proposed model accounts for randomness in both spectral and spatial network configurations as well as blockage effects. The interference model builds off a new blockage model which captures the average number of obstacles that cause a complete link outage. Using numerical simulations, we validate the theoretical results and demonstrate how beam directionality and randomness in node configuration impact the accumulated interference at arbitrary locations of a mmWave network.
\begin{figure}[b]
\centering
\includegraphics[scale=0.75]{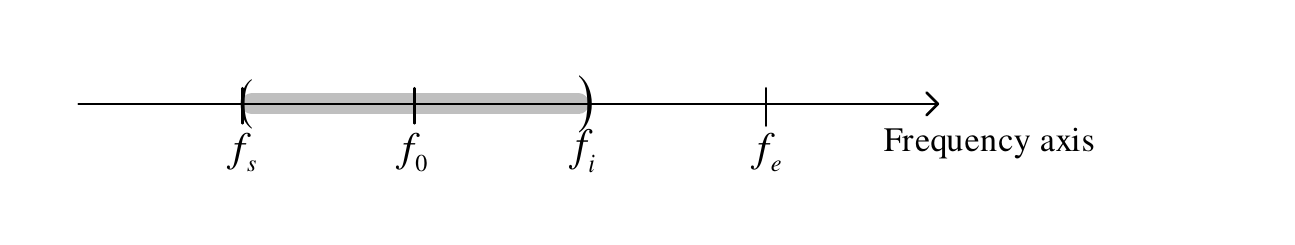}
\caption{The frequency axis.}
\label{fig:frequency_distance}
\end{figure}

\vspace{0.2cm}
\appendix
We assume that the sequence of frequencies used by interfering APs, i.e., ${f_i}$, are uncorrelated. However, the sequence of the frequency distance between the frequency of the victim receiver, $f_0$, and frequency of the $i^{\rm{th}}$ interfering AP, $f_i$, i.e., ${w_i=f_i-f_0}$, are correlated due to the common factor $f_0$. The conditional probability density function (PDF) of ${w_i=f_i-f_0}$ is given by the following lemma:
\begin{lemma} \label{lem:lemma4}
The conditional PDF of ${w_i=f_i-f_0}$ given $f_0$ is
\begin{align}
&{f_\Omega }\left( \omega  \right) {=} \left\{ \begin{array}{l}
\frac{2}{{{f_e} - {f_s}}}\,\,\,\,\,\,\,\,\,\,\,\,\,\,0 < \omega  \le \min \left( {\left| {\omega_{e}} \right|,\left| {\omega_{s}} \right|} \right)\\
\frac{1}{{{f_e} - {f_s}}}\,\,\,\,\min \left( {\left| {\omega_{e}} \right|,\left| {\omega_{s}} \right|} \right) < \omega  \le \max \left( {\left| {\omega_{e}} \right|,\left| {\omega_{s}} \right|} \right)
\end{array} \right.
\end{align}
where $\omega_{e}={f_e} - {f_0}$ and $\omega_{s}={f_s} - {f_0}$. We drop the subscript $i$ for notational simplicity.
\end{lemma}
\begin{proof}
Considering the frequency axis given in Fig.~\ref{fig:frequency_distance}, the frequency distance between the reference receiver and an interfering AP is calculated,\\
1) When $\omega \le \min \left( {\left| {\omega_{e}} \right|,\left| {\omega_{s}} \right|} \right)$, the CDF of ${\omega}$, i.e., ${F_{\Omega}}\left( \omega  \right)$, is the intersection of line segment $2\left|\omega\right|$ and $\left|{f_e}-{f_s}\right|$ divided by $\left|{f_e}-{f_s}\right|$.\\
2) When $\min \left( {\left| {\omega_{e}} \right|,\left| {\omega_{s}} \right|} \right) < \omega  \le \max \left( {\left| {\omega_{s}} \right|,\left| {\omega_{s}} \right|} \right)$, the CDF is $\frac{{\left| {\min \left( {\left| {{\omega _e}} \right|,\left| {{\omega _s}} \right|} \right)} \right| + {\omega }}}{{\left| {{f_e} - {f_s}} \right|}}$.
\begin{align} \label{eq:freq_dist_CDF} \notag
&{F_\Omega }\left( {{\omega }} \right) \\
&\hspace{-0.1cm}{=} \Bigg\{ \begin{array}{l} \hspace{-0.2cm}
\frac{{2{\omega }}}{{\left| {{f_e} - {f_s}} \right|}}\,\,\,\,\,\,\,\,\,\,\,\,\,\,\,\,\,\,\,\,\,\,\,\,\,\,\,\,\,\,\,\,\,\,\,\,\,\,\,\,\,\,\,\,\,\,\,0 < \omega  \le \min \left( {\left| {{\omega _e}} \right|,\left| {{\omega _s}} \right|} \right)\\
\hspace{-0.2cm}\frac{{\left| {\min \left( {\left| {{\omega _e}} \right|,\left| {{\omega _s}} \right|} \right)} \right| + {\omega }}}{{\left| {{f_e} - {f_s}} \right|}}\,\,\,\,\,\min \left( {\left| {{\omega _e}} \right|,\left| {{\omega _s}} \right|} \right) < {\omega}  \le \max \left( {\left| {{\omega _e}} \right|,\left| {{\omega _s}} \right|} \right)
\end{array}
\end{align}
Subsequently, the PDF of the frequency distance is given by taking derivative of the CDF with respect to $\omega$.
\end{proof}

\vspace{0.4cm}
\bibliographystyle{IEEEtran}

\bibliography{IEEEabrv,GBbibfile}

\end{document}